\pdfoutput=1
\documentclass[letterpaper, 10pt, conference]{ieeeconf}

\IEEEoverridecommandlockouts
\usepackage{amsmath,amssymb}

\usepackage{amsthm}
\usepackage{graphicx}
\usepackage{cite}
\usepackage[hidelinks]{hyperref}

\usepackage{setspace}
\everymath\expandafter{\the\everymath \thickmuskip=2mu \medmuskip=1mu}
\theoremstyle{plain}
\newtheorem{theorem}{Theorem}
\newtheorem{lemma}{Lemma}
\newtheorem{corollary}{Corollary}
\theoremstyle{definition}
\newtheorem{definition}{Definition}

\newcommand{\R}{\mathbb{R}}
\newcommand{\Rge}{\mathbb{R}_{\ge0}}
\newcommand{\U}{\mathcal{U}}
\newcommand{\Cs}{\mathcal{C}_{\mathrm{S}}}
\newcommand{\cA}{\mathcal{A}}
\newcommand{\RA}{\mathcal{R}_{\cA}}
\newcommand{\HA}{H_{\cA}}
\newcommand{\MA}{M_{\cA}}
\newcommand{\KA}{K_{\cA}}
\newcommand{\OT}{\mathcal{O}_{T}}
\newcommand{\kb}{\mathbf{k}_{\mathrm{b}}}
\newcommand{\vb}{f_{\mathrm{b}}}
\newcommand{\phib}{\varphi_{\mathrm{b}}}
\newcommand{\Phib}{\varPhi_{\mathrm{b}}}
\newcommand{\kp}{\mathbf{k}_{\mathrm{p}}}
\newcommand{\kQ}{\mathbf{k}^{\star}_{\cA}}
\newcommand{\WQ}{\mathbf{W}_{\mathrm{Q}}}
\newcommand{\lamA}{\lambda_{\cA}}
\newcommand{\epsA}{\varepsilon_{\cA}}
\newcommand{\nh}{n_{h}}
\newcommand{\nlink}{n_{\mathrm{link}}}
\newcommand{\Tch}{T_{\mathrm{ch}}}
\newcommand{\albhat}{\widehat{\alpha}}
\newcommand{\alb}{\alpha_{\mathrm{b}}}
\newcommand{\appext}{\preceq_{\mathrm{app}}}
\newcommand{\dist}{\operatorname{dist}}
\newcommand{\defeq}{\triangleq}
\DeclareMathOperator*{\argmin}{arg\,min}

\begin{document}

\title{ \bf
Backup Control Barrier Functions with Memory \\ for Online Set Expansion
}

\author{ M. Yusuf Uzun and Ersin Da\c{s}
\thanks{ M. Yusuf Uzun and Ersin Da\c{s}
    are with the Department of Mechanical, Materials, and Aerospace Engineering, Illinois Institute of Technology, Chicago, IL 60616, USA, {\texttt{muzun@hawk.illinoistech.edu}, \texttt{edas2@illinoistech.edu}}.}
}

\maketitle
\thispagestyle{plain}
\pagestyle{plain}

\begin{abstract}
Backup control barrier functions (bCBFs) construct implicit safe sets for nonlinear systems with bounded inputs using finite-horizon predictions under a backup controller. As in model predictive control, an invariant backup set serves as a terminal set that guarantees safe continuation beyond the prediction horizon. However, reaching a small terminal set may require long prediction horizons, increasing the cost of online prediction. In this paper, we introduce an online construction that builds a certified region from sets stored along backup trajectories while keeping the prediction horizon fixed. The method stores finite collections of balls and connects each newly certified collection to the backup set or to previously stored balls. We derive conditions that guarantee forward invariance of the union of these balls and the backup set under a fixed backup controller. These conditions also ensure that every stored state reaches the backup set in finite time. The resulting union serves as an enlarged terminal set, allowing safety to be certified with shorter prediction horizons. We formulate a feasible quadratic program for safe control and prove that safety is preserved as additional certified collections are stored online. Simulations on a quadrotor demonstrate safe task completion with a shorter prediction horizon than the standard bCBF.
\end{abstract}

\section{Introduction}

Safety-critical control of autonomous systems requires controllers that satisfy safety constraints while respecting input bounds.
Control barrier functions (CBFs) provide a systematic approach to this problem by enforcing the forward invariance of a prescribed safe set \cite{ames2017}.
However, verifying that a safety function is a CBF can be challenging for nonlinear systems with bounded inputs, since an admissible safe input may not exist at every state in the set.

Backup control barrier functions (bCBFs) construct online a subset of the safe set in which an admissible safe input always exists, using finite-horizon predictions under a backup controller, without computing a maximal controlled invariant set \cite{gurriet2020}.
Starting from a backup set that is forward invariant under this controller, the implicit safe set consists of states from which the backup controller keeps the system safe and reaches the backup set within the prediction horizon.
Similar to model predictive control (MPC), the backup set serves as a terminal set in which safety can be maintained after the prediction horizon. The forward invariance of the implicit safe set can then be enforced through a quadratic program (QP) whose feasibility follows from the backup controller.
However, for a fixed terminal set, a short prediction horizon may yield a conservative implicit safe set. Increasing the prediction horizon can enlarge this set, but also increases the cost of online prediction. This motivates enlarging the terminal set while keeping the prediction horizon fixed.

Several extensions aim to enlarge the safe sets certified by backup-based safety filters \cite{kim2026}. These include aggregating multiple backup controllers \cite{ong2026}, separating the set-expanding controller from the backup controller \cite{vanwijk2026}, and using multiple backup policies with a state-dependent reachable horizon \cite{hwang2026}. Related approaches expand controlled invariant sets through linear-program certificates \cite{vahs2025} or approximate regions of attraction offline through flow-control barrier functions and constrained polytope expansion \cite{ubellacker2024}. Feasibility guarantees have also been developed for smoothed bCBFs \cite{alan2026}. Within predictive safety filter design \cite{wabersich2021}, backup trajectories have been used to extend safe sets during receding-horizon optimization \cite{tearle2021}, while online trajectory verification can retain a previously verified plan until a new plan is certified \cite{agrawal2024}. The retention and composition of certified trajectories also appear in sequential composition \cite{burridge1999}, funnel trees and libraries \cite{tedrake2010,majumdar2017}, and learning model predictive control \cite{rosolia2018}.

In this paper, we propose to enlarge the invariant terminal set online. The construction stores certified collections of balls, called links, along backup trajectories. Each link connects to the backup set or to a previously stored ball, so that every chain of links ends in the backup set. The union of the stored balls and the backup set serves as the enlarged terminal set. A prediction can therefore terminate in a stored ball, while the subsequent trajectory under the same backup controller reaches the backup set through the stored chain. The total time to the backup set can thus exceed the prediction horizon. The main contributions are as follows.
\begin{enumerate}
\item We establish sufficient conditions that guarantee forward invariance of the enlarged terminal set and finite time arrival at the backup set from every stored ball under the backup controller.
\item We prove that links can always be built along a finite portion of a backup trajectory that stays strictly inside the safe set and ends in the interior of the backup set or of a stored ball. The new links strictly enlarge the terminal set when the trajectory starts outside it.
\item We formulate a feasible, QP-based safety filter and establish safety and global existence of closed loop solutions under certified online additions, with finitely many updates in each bounded time interval.
\end{enumerate}

We organize this paper as follows. Section~\ref{sec:setting} provides preliminaries. Section~\ref{sec:growth} presents the terminal set construction.
Section~\ref{sec:filter} develops the continuous-time safety filter, and Section~\ref{sec:appends} describes online terminal set growth. Section~\ref{sec:sim}
presents a quadrotor example, and Section~\ref{sec:conclusion} concludes.

\section{Preliminaries}
\label{sec:setting}

\emph{Notation:} For a matrix ${\mathbf{M}  = \mathbf{M} ^{\top} \succ 0}$ and a vector ${v}$ of matching size, we write ${\| v \|_{\mathbf{M}} \defeq (v^{\top} \mathbf{M}  v)^{1/2}}$. Throughout, ${\mathbf{P} = \mathbf{P}^{\top} \succ 0}$ is a fixed ${n \times n}$ matrix, and ${\| \cdot \| \defeq \| \cdot \|_{\mathbf{P}}}$ denotes the norm on ${\R^n}$. For a center ${c \in \R^n}$ and radius ${r \geq 0}$, ${B(c, r) \defeq \{ x \in \R^n : \| x - c \| \leq r \}}$ denotes the closed ball. The function ${\dist(x, S) \defeq \inf_{p \in S} \| x - p \|}$ gives the distance from a point ${x}$ to a set ${S}$, ${\dist(K, S) \defeq \inf_{x \in K} \dist(x, S)}$ from a set ${K}$, and ${\operatorname{int} S}$ denotes the interior of ${S}$.

We consider a nonlinear control-affine system of the form
\begin{equation}\label{eq:system}
\dot x=f(x)+g(x)u,
\end{equation}
where ${x \in \R^n}$, ${u \in \U}$, ${f: \R^n \to \R^n}$ and ${g : \R^n \to \R^{n \times m}}$ are $C^1$, and the admissible input set ${\U \subset \R^m}$ is a convex polytope.

\subsection{Control Barrier Functions}
To describe safety, we define the safe set
\begin{equation*}
\Cs \defeq \{x \in \R^n : h_i(x) \ge 0, \ i = 1, \dots, \nh \},
\end{equation*}
with $C^1$ safety functions ${h_i : \R^n \to \R}$, and with interior ${\operatorname{int}\Cs = \{x : \min_i h_i(x) > 0\}}$.
The control objective is to keep the state in $\Cs$ while satisfying the input constraint ${u \in \U}$. Control barrier functions (CBFs) \cite{ames2017} provide conditions on the control input that render $\Cs$ forward invariant.
The functions ${h_i}$ are compatible CBFs \cite{tan2022} if there exist extended class-${\mathcal{K}}$ functions ${\alpha_i}$ such that ${\forall x \in \Cs}$,
\begin{equation}
\label{eq:setting:cbf-feasibility}
\sup_{u \in \U} ~ \min_{i = 1, \dots, \nh} \!
\bigl[ \nabla h_i(x) \bigl( f(x) + g(x)u \bigr)
+ \alpha_i \bigl( h_i(x) \bigr) \bigr] \!>\! 0.
\end{equation}
This condition requires a common admissible input for all safety functions at each state. Then, any locally Lipschitz controller ${\mathbf{k} : \R^n \to \U}$ satisfying
\begin{equation}
\label{eq:setting:cbf-control}
\!\!\nabla h_i(x) \bigl( f(x) + g(x)\mathbf{k}(x) \bigr) \!\geq\! - \alpha_i \bigl( h_i(x) \bigr),  \  i \!=\! 1, \dots, \nh,
\end{equation}
${\forall x \in \Cs}$ renders ${\Cs}$ forward invariant.

\subsection{Backup Control Barrier Functions}
While CBFs provide a systematic way to synthesize safe controllers, arbitrary safety functions ${h_i}$ may not satisfy the CBF condition \eqref{eq:setting:cbf-feasibility}, especially with bounded inputs. In this case, there may exist states in ${\Cs}$ for which no admissible input satisfies all the constraints in \eqref{eq:setting:cbf-control}. Backup CBFs (bCBFs) \cite{gurriet2020} address this issue by constructing a subset of ${\Cs}$ where safe inputs are guaranteed to exist and by generating corresponding safe controllers within the CBF framework. In the rest of this paper, the functions $h_i$ define the safe set and are not assumed to satisfy \eqref{eq:setting:cbf-feasibility}.

The \textit{backup controller} is a $C^1$ map ${\kb : \R^n \to \U}$ that yields the closed loop dynamics
\begin{equation}
\label{eq:setting:backup}
\dot x = \vb(x) \defeq f(x) + g(x) \kb(x).
\end{equation}
To ensure that solutions of \eqref{eq:setting:backup} are defined for all time, we assume that there exist constants ${G_0, G_1 \geq 0}$ such that ${\|\vb(x)\|\le G_0+G_1\|x\|}$ for all ${x \in \R^n}$.
The flow $\phib(t,x)$ denotes the state reached at time $t$ under $\kb$, starting from $x$ at time zero. Since $\vb$ is $C^1$ and satisfies this growth bound, the flow is defined for all ${(t,x) \in \R \times \R^n}$ and
is jointly $C^1$. It satisfies
\begin{gather}
\phib(0,x)=x,\qquad\partial_t\phib(t,x)=\vb\bigl(\phib(t,x)\bigr),
\label{eq:setting:complete}\\
\phib(t_1+t_2,x)=\phib\bigl(t_1,\phib(t_2,x)\bigr), \quad
\forall t_1,t_2\in\R .\label{eq:setting:semigroup}
\end{gather}
Then, differentiating \eqref{eq:setting:semigroup} in $t_2$ at $0$ gives \cite{chen2021}
\begin{equation}
\label{eq:setting:flowid}
\Phib(\tau, x) \vb(x) = \vb\bigl(\phib(\tau,x) \bigr),
\qquad \forall (\tau, x) \in \R \times \R^n,
\end{equation}
where ${\Phib(t,x) \defeq \partial\phib(t,x) / \partial x}$ is the sensitivity of the flow to its initial state.

To establish forward invariance, we adapt the comparison lemma to the maximum of finitely many $C^1$ functions.%
\begin{lemma}
\label{lem:nocross}
Consider finitely many functions $w_i : [t_0,t_1] \to \R$ of class $C^1$ with maximum ${w \defeq \max_i w_i}$ such that
${w(t_0) \geq 0}$. Suppose that, for constants ${\varepsilon > 0}$ and ${\lambda \in \R}$,
\begin{equation}
\label{eq:nocross:diff}
            \dot w_i(t)\ge\lambda\,w(t)
\end{equation}
holds for all ${t \in [t_0, t_1]}$ such that
${w(t) \in (-\varepsilon, 0]}$, for at least one $i$ with ${w_i(t) = w(t)}$. Then, ${w(t) \ge 0}$ for all
${t \in [t_0, t_1]}$.
\end{lemma}

\begin{proof}
Suppose, for contradiction, that ${w(t_a) < 0}$ for some
${t_a \in (t_0,t_1]}$. Let ${t_b \defeq \max \{t \in [t_0, t_a] : w(t) \geq 0\}}$
and ${t_c \in (t_b,t_a]}$ be the first time with
${w \leq -\varepsilon}$, or $t_a$. On $[t_b,t_c)$ we have ${w(t_b) = 0}$ and ${-\varepsilon < w \le 0}$, while ${w(t_c) < 0}$. By
\eqref{eq:nocross:diff}, at any $t\in[t_b,t_c)$ some $i$ with $w_i(t)=w(t)$ has $\dot w_i(t)\ge\lambda w(t)$, so that ${z_i(t) \defeq e^{-\lambda t} w_i(t)}$ has ${\dot z_i(t) \geq 0}$. For ${\sigma > 0}$ satisfying ${e^{-\lambda t_c} w(t_c) + \sigma (t_c - t_b) < 0}$, we define ${Y(t) \defeq e^{-\lambda t} w(t) + \sigma (t - t_b)}$ and ${t_d \defeq \max \{t \in [t_b, t_c] : Y(t) \geq 0\}}$, so that ${t_d < t_c}$ and ${Y < Y(t_d)}$ on ${(t_d, t_c]}$. For such an $i$ at $t_d$ and ${t \in (t_d, t_c]}$,
\begin{equation*}
    \begin{aligned}
        z_i(t) & \leq e^{-\lambda t} w(t) = Y(t) - \sigma (t - t_b) \\
        & < Y(t_d) - \sigma (t - t_b) = z_i(t_d) - \sigma (t - t_d) ,
    \end{aligned}
\end{equation*}
so ${\dot z_i(t_d) \leq -\sigma < 0}$, contradicting ${\dot z_i(t_d) \geq 0}$.
\end{proof}

The backup construction \cite{gurriet2020} requires a set ${S_0 \subseteq \Cs}$ that is forward invariant under $\kb$, which we specify through a defining function and an inequality on its derivative along~\eqref{eq:setting:backup}.
\begin{definition}
	\label{def:base}
	A $C^1$ function ${b_0 : \R^n \to \R}$ and constants ${\varepsilon_0 > 0}$ and ${\lambda_0 \geq 0}$ define a \emph{base set} ${S_0 \defeq \{x \in \R^n : b_0(x) \geq 0\}}$ if $S_0$ is nonempty, compact and contained in $\Cs$, with interior ${\operatorname{int} S_0 = \{x : b_0(x) > 0\}}$, and
\begin{equation}
    \label{eq:setting:basedrift}
    \nabla b_0(x)\,\vb(x) \geq -\lambda_0\,|b_0(x)|,
\end{equation}
	for all $x$ satisfying ${|b_0(x)| < \varepsilon_0}$.
\end{definition}
\looseness=-1
The backup controller renders a base set $S_0$ forward invariant. For ${x \in S_0}$, \eqref{eq:setting:complete}, \eqref{eq:setting:basedrift} and Lemma~\ref{lem:nocross} with the single function ${w(t) \defeq b_0(\phib(t,x))}$, ${\varepsilon = \varepsilon_0}$ and ${\lambda = \lambda_0}$ on every $[0,t_1]$ give ${\phib(t,x)\in S_0}$ for all ${t\ge0}$. Note that condition~\eqref{eq:setting:basedrift} bounds the decrease of $b_0$ near zero without requiring attraction from outside $S_0$. For a compact set with a smooth boundary, forward invariant under $\kb$, such a $b_0$ can be constructed from the signed distance (positive inside, negative outside) by smooth saturation away from the boundary. The bound then holds for sufficiently small $\varepsilon_0$, with $\lambda_0$ a Lipschitz constant of $\vb$ near $S_0$. For a quadratic Lyapunov sublevel set with strict decrease of the Lyapunov function on the boundary, \eqref{eq:setting:basedrift} holds with $\lambda_0 = 0$ for sufficiently small $\varepsilon_0$ \cite{gurriet2020}.

The safe controller is designed to keep the state in the operating set of a terminal set: the states whose trajectory under $\kb$ over a fixed prediction horizon ${T > 0}$, the \emph{rollout}, stays safe and ends in that terminal set. A \textit{backup arc} is a finite portion of a backup trajectory, a solution of \eqref{eq:setting:backup}.
\begin{definition}
\label{def:opset}
For a closed set $\mathcal R\subseteq\Cs$, the \emph{operating set} is
\begin{equation}
\label{eq:setting:opset}
    \begin{aligned}
        \OT(\mathcal R) \defeq \bigl \{x \in \R^n :{} &\phib(T,x) \in \mathcal R,\
        h_i(\phib(\tau,x)) \geq 0,\\
        & \forall i\in\{1,\dots,\nh\},\ \forall \tau \in [0, T] \bigr\}.
    \end{aligned}
\end{equation}
\end{definition}
For ${\mathcal R = S_0}$, it is the implicit safe set, the $T$-time constrained reachable set of \cite{chen2021}. For closed sets ${\mathcal R \subseteq \mathcal R' \subseteq \Cs}$, $\OT(\mathcal R)$ is closed, as an intersection of preimages of closed sets under continuous maps, and ${\OT(\mathcal R)\subseteq\OT(\mathcal R')}$ by \eqref{eq:setting:opset}.

\section{Certified Terminal Sets}
\label{sec:growth}
To enlarge the operating set with a fixed prediction horizon, we construct a terminal set by adding finite collections of balls to $S_0$. We derive conditions for forward invariance and finite time arrival in $S_0$, and then construct collections satisfying these conditions along backup arcs.

\subsection{Links and Certificates}

\begin{figure}[t]
\centering
\setlength{\abovecaptionskip}{0.5pt}
\includegraphics[width=\linewidth]{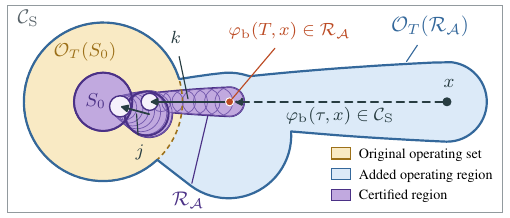}
\caption{Growth of the certified region \(\mathcal R_{\mathcal A}\) by links \(j\) and \(k\). Their final balls (hollow circles) lie inside \(S_0\) and a ball of \(j\), respectively. The blue boundary encloses the enlarged operating set \(\mathcal O_T(\mathcal R_{\mathcal A})\). The dashed rollout carries \(x\) into \(\mathcal R_{\mathcal A}\) within the prediction horizon \(T\).}
\label{fig:schematic}
\vspace{-5mm}
\end{figure}

A \textit{link} is a
finite ordered sequence of stored balls. Under the backup controller, each ball is transferred to the next, and the last ball is contained in the base set or in a ball of an earlier link (see Fig.~\ref{fig:schematic}). These connections allow new links to extend a previously constructed terminal set.

A \emph{family} $\cA$ consists of ${\nlink \geq 0}$ links, indexed by
${k= 1, \dots, \nlink}$ in the order they are appended. Link $k$ contains $N_{k}+1$ stored balls, with $N_{k} \geq 1$.
We index the base set $S_0$ by $0$ and the ball at position $l$ of link $k$ by the pair $(k,l)$, ${l = 0,\dots, N_k}$. We write $a$ for an index, and for ${a \neq 0}$ the \emph{stored ball} is $S_a\defeq B(c_a,R_a)$ with $R_a>0$. The index set is $\cA\defeq \{0\} \cup \{(k,l) : k=1,\dots,\nlink,\ l=0,\dots,N_k\}$. The base set and the balls of links with smaller indices are called \emph{older} than the balls of link $k$.
The \emph{successor} map $p(a)$, $p:\cA\setminus\{0\}\to\cA$, carries each \emph{flow step} $(k,l)$, $l < N_k$, to $(k,l+1)$ with duration $\vartheta_{(k,l)}>0$, the durations of link $k$ summing to $L_k$. It carries the \emph{dock} $(k,N_k)$ to $0$ or to a ball of an older link, so successor chains reach $0$ after at most $\nlink$ docks. A dock $a$ with $p(a)=0$ is a \emph{base dock}. The first ball $(k,0)$ of link $k$ is its \emph{entrance}.

The \emph{certified region} associated with $\cA$ (see Fig.~\ref{fig:schematic}) is ${\RA\defeq S_0\cup\bigcup_{a\ne0}S_a}$. To use this union as a terminal set, the balls must lie in $\Cs$, and the backup controller must keep trajectories in $\RA$ while transferring them to their successors.

\begin{definition}
\label{def:semfam}
A family $\cA$ is a \emph{ball family} if it satisfies the following conditions. \\ \emph{Safety:} Every stored ball lies in $\Cs$. \\ \emph{Arrival:} Every flow step $a$ satisfies $\phib(\vartheta_a,S_a)\subseteq S_{p(a)}$. \\ \emph{Exposure:} For each flow step $a$, with chord $d_a \defeq c_{p(a)} - c_a$, the successor radius satisfies
\begin{equation}
\label{eq:succ:exposure}
R_{p(a)}^{2}\ge(1+2\mu_a\vartheta_a)R_a^{2}+\|d_a\|^{2}+2R_aE_a
\end{equation}
with $E_a\defeq\|d_a-\vartheta_a\vb(c_a)\|$, for some thickness $\eta_a > 0$ and a one-sided Lipschitz constant $\mu_a\in\R$ satisfying\begin{equation}\label{eq:succ:osl}
    \begin{aligned}
        (x-c_a)^{\top}\mathbf{P}&\bigl(\vb(x)-\vb(c_a)\bigr)\le\mu_a\|x-c_a\|^{2}\\
        & \forall x: \|x-c_a\|<R_a+\eta_a .
    \end{aligned}
\end{equation}
\emph{Docks:} Every dock $a$ with $p(a)\ne0$ satisfies $S_a \subseteq S_{p(a)}$, with
equality allowed.
Every base dock satisfies, for some ${\eta_a > 0}$, ${R_a^2 - \|x - c_a\|^2 \leq R_a b_0(x)}$ for all $x$ with ${\|x - c_a\| < R_a + \eta_a}$. \emph{Horizon:} The total duration of link $k$ satisfies ${L_k \in (0,T]}$.
\end{definition}

The arrival condition specifies the transfer at time $\vartheta_a$. The exposure condition prevents outward flow through the portion of ${\partial S_a}$ outside $\operatorname{int}S_{p(a)}$, keeping trajectories in $\RA$ during the transfer.
A negative $\mu_a$ reduces the required successor radius, so a link may narrow where the flow contracts.
Near an equilibrium with a Hurwitz linearization $\mathbf{A}$, a matrix $\mathbf{P}$ satisfying ${\mathbf{A}^{\top}\mathbf{P}+\mathbf{P}\mathbf{A}\prec0}$ exists and gives a negative $\mu_a$ for balls in a neighborhood of the equilibrium.
For ${\mathbf{P}=\mathbf{I}}$ this condition becomes ${\mathbf{A} + \mathbf{A}^{\top} \prec 0}$. We therefore choose $\mathbf{P}$ together with the backup controller.
This choice gives a constant contraction metric in that neighborhood \cite{aminzare2014}.

For $c\in\R^n$ and $r>0$ the \emph{ball function}
\begin{equation*}
\beta_{c,r}(x)\defeq {\big(r^{2}-\|x-c\|^{2} \big)}/{r}
\end{equation*}
is nonnegative on $B(c,r)$, and is $C^1$.
For a family $\cA$, set $\beta_0\defeq b_0$ and $\beta_a\defeq\beta_{c_a,R_a}$ for $a\ne0$. The \emph{terminal envelope} is ${\HA(x)\defeq\max_{a\in\cA}\beta_{a}(x)}$,
and, as $\beta_a\ge0$ on $S_a$, the certified region is
\begin{equation}\label{eq:env:level}
\RA=\{x\in\R^n : \HA(x)\ge0\} .
\end{equation}
The \emph{active set} $\MA(x)\defeq\{a\in\cA : \beta_{a}(x)=\HA(x)\}$ is nonempty for all ${x \in \R^n}$, as $\cA$ is finite and contains $0$. The envelope $\HA$ is continuous as a maximum of finitely many $C^1$ functions. We compare the ball functions of nested balls.
\begin{lemma}
\label{lem:envelope}
If ${B(c_1, r_1) \subseteq B(c_2, r_2)}$ with ${r_1, r_2 > 0}$, then
\begin{equation}\label{eq:env:nest}
    \beta_{c_{2},r_{2}}(x)\ge\beta_{c_{1},r_{1}}(x)
    \qquad\forall x\in\R^n.
\end{equation}
\end{lemma}

\begin{proof}
The ball $B(c_1,r_1)$ contains a point at distance ${\|c_1 - c_2\| + r_1}$ from $c_2$, hence ${\|c_1 - c_2\| \leq r_2 - r_1}$ and ${\|x - c_2\| \leq \|x - c_1\| + r_2 - r_1}$. Expanding the squares gives ${\beta_{c_2,r_2}(x) -\beta_{c_1,r_1}(x) \geq \frac{(r_2-r_1)(r_1-\|x-c_1\|)^2}{r_1r_2} \geq 0}$.
\end{proof}

Lemma~\ref{lem:envelope} is used at the docks whose successor is a stored ball. Such a dock lies inside its successor, whose ball function is at least as large everywhere, so the successor can be used in place of the dock when we show that backup trajectories do not leave the certified region. This is why Definition~\ref{def:semfam} requires no exposure condition of these docks.

\subsection{Invariance, Chain Time and Append Monotonicity}

For a family $\cA$ and ${a \in \cA}$, the \emph{chain time} $\Tch(a)$ is the sum of the flow step durations along its successor chain to $0$, given by ${\Tch(0) = 0}$ and ${\Tch(a) = \vartheta_a + \Tch(p(a))}$ for ${a \neq 0}$, with ${\vartheta_a = 0}$ at a dock. The recursion terminates because successor chains reach $0$.

An \emph{append extension} $\cA\appext\cA'$ of a family $\cA$ with $\nlink$ links is a ball family $\cA'$ whose first links $1,\dots,\nlink$ are those of $\cA$, each with the flow step count, centers, radii, durations, and successor it has in $\cA$. This requirement gives $\cA\subseteq\cA'$ as index sets with the same balls on $\cA$.

To establish forward invariance of $\RA$, we apply Lemma~\ref{lem:nocross} to the functions $\beta_a$ along the backup flow.
For a ball family, we choose $(\mu_a,\eta_a)$
for each flow step $a$ to satisfy the exposure condition, and $\eta_a$ for each base dock $a$ to satisfy the docks condition (see Definition~\ref{def:semfam}). We derive the bound that Lemma~\ref{lem:nocross} requires from these data through the exposure inequalities \eqref{eq:succ:exposure} of the flow steps. We define
\begin{subequations}
\begin{align}\label{eq:collar:const}
    \ell_a &\defeq 2\mu_a -({R_{p(a)}-R_a})/({R_a\vartheta_a})+ {E_a}/({R_a\vartheta_a}),\\
    \lamA&\defeq\max\bigl(\{\lambda_0\}\cup\{|\ell_a|\}\bigr),\label{eq:collar:lam}\\ \label{eq:collar:eps}
    \epsA&\defeq\min\bigl(\{\varepsilon_0\}\cup
    \{2\eta_a+\eta_a^2/R_a\}\bigr),
\end{align}
\end{subequations}
where $a$ ranges over the flow steps in $\ell_a$ and $\lamA$ and over the flow steps and the base docks in $\epsA$. The docks condition (see Definition~\ref{def:semfam}) requires ${b_0 \geq \beta_a}$ near a base dock $a$, which ${S_a \subseteq S_0}$ and \eqref{eq:env:nest} give if $b_0$ is a ball function.

\begin{lemma}\label{lem:collar}
Suppose $S_0$ is a base set and $\cA$ is a ball family. Then, $\epsA>0$ and $\nabla\beta_a(x)\,\vb(x)\ge-\lamA\,|\HA(x)|$ for some $a\in\MA(x)$ if $-\epsA<\HA(x)<\varepsilon_0$.
\end{lemma}

\begin{proof}
See Appendix~\ref{app:collar}.
\end{proof}

Having bounded the decrease of the terminal envelope near zero, we now establish the invariance and growth properties of the certified region.

\begin{theorem}\label{thm:growth}
Consider the system \eqref{eq:system}. Suppose $S_0$ is a base set and $\cA$ is a ball family with $\nlink$ links. Then, $\RA$ is a nonempty compact subset of $\Cs$, and $\kb$ renders it forward invariant. Consequently, ${\RA \subseteq \OT(\RA)}$. In addition, for any ${a \in \cA}$ and ${x \in S_a}$, ${\phib(t,x) \in S_0}$,  ${\forall t \geq \Tch(a)}$, where ${\Tch(a) \leq \nlink\,T}$.  Moreover, every append extension ${\cA \appext \cA'}$ satisfies ${\RA \subseteq \mathcal R_{\cA'}}$ and ${\OT(\RA) \subseteq \OT(\mathcal R_{\cA'})}$.
\end{theorem}

\begin{proof}
See Appendix~\ref{app:growth}.
\end{proof}
By Theorem~\ref{thm:growth}, a backup trajectory that enters a stored ball remains in $\RA$ and reaches the base set within the chain time of that ball, which can exceed the prediction horizon. It remains to construct links that extend a ball family.

\subsection{Certified Links along Backup Arcs}

To extend a family $\cA_{\mathrm{old}}$, we take a backup arc ${c : [0,L] \to \R^n}$, ${L > 0}$, and divide $[0,L]$ into ${n_{\mathrm{s}} \defeq \lceil L/T\rceil}$ equal segments. For a radius profile ${r : [0,L] \to (0,\infty)}$ and an integer ${N \geq 1}$, we place balls $B(c(s),r(s))$ at the $N+1$ equally spaced values of $s$ in each segment, including both endpoints. The balls of each segment, ordered by increasing $s$, form a link of $N$ flow steps of duration $L/(n_{\mathrm{s}}N)$. The links are appended in order of decreasing $s$. The dock at $c(L)$ has successor ${a_L \in \cA_{\mathrm{old}}}$. Each remaining dock coincides with the entrance of the next link along the arc, and that entrance is its successor.

\begin{theorem}\label{thm:link}
Consider the system \eqref{eq:system}. Suppose $S_0$ is a base set, $\cA_{\mathrm{old}}$ is a ball family, ${c([0,L]) \subseteq \operatorname{int} \Cs}$ and ${c(L) \in \operatorname{int} S_{a_L}}$. Then, there exist a radius profile $r$ and an integer ${N_0 \geq 1}$ such that, for every ${N \geq N_0}$, appending the links placed along $c$ to $\cA_{\mathrm{old}}$ gives an append extension of $\cA_{\mathrm{old}}$. Moreover, this append extension strictly enlarges the certified region if ${c(0) \notin \mathcal R_{\cA_{\mathrm{old}}}}$.
\end{theorem}

\begin{proof}
See Appendix~\ref{app:link}.
\end{proof}
Theorem~\ref{thm:link} establishes the existence of an append extension along a given backup arc. For safe controller synthesis, we next formulate a backup CBF safety filter with the resulting certified region as its terminal set.

\section{Continuous-Time Safety Filter}
\label{sec:filter}
We formulate the safety filter as a semi-infinite quadratic program (QP), then establish its feasibility, operating set invariance, and the existence of global closed loop solutions. The safety filter projects the nominal input onto constraints that preserve the operating set.

For a ball family $\cA$, Theorem~\ref{thm:growth} allows the safety filter to use the terminal condition $\phib(T,x)\in\RA$ in place of $\phib(T,x)\in S_0$. With $x_T\defeq\phib(T,x)$, the \emph{operating margin} is
\begin{equation}
\label{eq:filter:margin}
m_{\cA}(x)\defeq\min \Big\{\HA(x_T),\ \min_{\substack{i=1,\dots,\nh\\
\tau\in[0,T]}}h_i\bigl(\phib(\tau,x)\bigr) \Big\}.
\end{equation}
For ${\cA = \{0\}}$, $m_{\cA}$ is the backup CBF of \cite{chen2021}.

We use $s^{+}\defeq\max\{s,0\}$ and denote by $\KA(x)$ the set of $u\in\U$ satisfying the rollout constraints
\begin{equation}
\label{eq:filter:roll}
\begin{aligned}
    &\!\!\nabla h_i(\phib(\tau,x))\Phib(\tau,x)g(x)\bigl(u-\kb(x)\bigr)\\
    &\!\!\geq-\alpha_i\bigl(h_i(\phib(\tau,x))^{+}\bigr), \ \forall i\in\{1,\dots,\nh\},\ \forall \tau \in [0, T],
\end{aligned}
\end{equation}
and the terminal constraints
\begin{equation}\label{eq:filter:term}
\begin{aligned}
    \nabla\beta_a(x_T)&\Phib(T,x)g(x)\bigl(u-\kb(x)\bigr)\\
    &\ge-\albhat_a\bigl(\HA(x_T)^{+}\bigr)-\rho\bigl(\HA(x_T)-\beta_a(x_T)\bigr)
\end{aligned}
\end{equation}
for all $a\in\cA$, where $\albhat_0 \defeq \alpha_0$ and $\albhat_a(\zeta) \defeq \alb(R_a\zeta)/R_a$ for each stored ball $a$ of radius $R_a$. The gains $\alpha_i$, $\alb$ and $\alpha_0$ are fixed independently of the family and are locally Lipschitz extended class-$\mathcal K$ functions \cite{ames2017}. The \emph{gap relaxation} $\rho:\Rge \to \Rge$ is continuous with $\rho(0) = 0$.

For a continuous nominal controller $\kp:\Rge\times\R^n\to\R^m$ and a weight matrix $\WQ=\WQ^\top\succ0$, the \emph{safety filter} of $\cA$ is
\begin{equation}
\label{eq:filter:qp}
\begin{aligned}
    \kQ(t,x)  \defeq  \argmin_{u\in\U}\ &
    \bigl\|u-\kp(t,x)\bigr\|_{\WQ}^{2} \\
    \text{s.t.} \ \ &\eqref{eq:filter:roll},\ \ \forall i,\ \forall\tau\in[0,T],\\ & \eqref{eq:filter:term},\ \ \forall a\in\cA .
\end{aligned}
\end{equation}
The cost function of \eqref{eq:filter:qp} is quadratic, and its constraints are affine in $u$. Since $\U$ is a polytope, \eqref{eq:filter:qp} is a semi-infinite QP.
Every constraint depends on the input $u$ only through $g(x)\bigl(u-\kb(x)\bigr)$, so its left side is zero at $u=\kb(x)$. Its right side is nonpositive at every $x\in\R^n$ since the gains have the clipped arguments $h_i(\phib(\tau,x))^{+}$ and $\HA(x_T)^{+}$, equal to the unclipped ones on $\OT(\RA)$, and the nonnegative function $\rho$ has the \emph{envelope gap} $\HA(x_T)-\beta_a(x_T)$ as its argument. Hence $\kb(x)\in\KA(x)$ for all $x\in\R^n$, and $\KA(x)$, the intersection of $\U$ with closed half spaces, is nonempty, compact, and convex. Thus, \eqref{eq:filter:qp} is feasible $\forall (t,x)$ and has a unique minimizer since ${\WQ \succ 0}$. We remark that the relaxation $\rho$ acts only on terminal constraints with $a\notin\MA(x_T)$. Clipping the gains and omitting the terminal drift $\nabla\beta_a(x_T)\vb(x_T)$ ensure feasibility at any state, extending the feasibility guarantee of \cite{chen2021}. The closed loop system under $\kQ$ is
\begin{equation}
\label{eq:inv:qploop}
\dot x(t)=f(x(t))+g(x(t))\,\kQ(t,x(t)),\qquad x(t_0)=x_0,
\end{equation}
with ${t_0 \geq 0}$ and ${x_0 \in \R^n}$.
A \emph{solution} on an interval ${I \subseteq \Rge}$ containing $t_0$ is a $C^1$ map ${x : I \to \R^n}$ satisfying \eqref{eq:inv:qploop} on $I$.

\begin{theorem}\label{thm:invariance}
Consider the system \eqref{eq:system}. Suppose $S_0$ is a base set and $\cA$ is a ball family. Then, $\kQ$ renders ${\OT(\RA) \subseteq \Cs}$ forward invariant.
Moreover, every solution of \eqref{eq:inv:qploop} on a bounded interval $[t_0,t_1)$ with ${m_{\cA}(x(t_0)) > 0}$ satisfies ${\inf_{t \in [t_0,t_1)} m_{\cA}(x(t)) > 0}$, and $\kQ$ is continuous at every $(t,x)$ with ${m_{\cA}(x) > 0}$. Finally, the closed loop has a solution from every $x_0$ with ${m_{\cA}(x_0) > 0}$, and every maximal solution from such an $x_0$ is defined on $[t_0,\infty)$.
\end{theorem}

\begin{proof}
See Appendix~\ref{app:invariance}.
\end{proof}

Theorem~\ref{thm:invariance} establishes safety for a fixed ball family on its operating set, which contains the implicit safe set for the same prediction horizon. The family, however, grows online.

\section{Online Growth}
\label{sec:appends}
We first extend the safety and existence results to locally finite append sequences, then present the online growth rules.

\subsection{Locally Finite Appends}

An \emph{append path} ${t \mapsto \cA(t)}$, ${t \geq t_0 \geq 0}$, equals $\cA_j$ on $[t_j,t_{j+1})$, where $\cA_0$ is a ball family, ${\cA_{j-1} \appext \cA_j}$ for ${j \geq 1}$, and the \emph{update times} ${t_1 < t_2 < \cdots}$ after $t_0$ are finitely many in every bounded interval, with ${t_{j+1} = \infty}$ after the last one. The \emph{grown safety filter} $\mathbf{k}^{\star}_{\cA(t)}(t,x)$ is the minimizer of \eqref{eq:filter:qp} for $\cA(t)$ and may jump at the update times. We consider the closed loop system
\begin{equation}\label{eq:append:loop}
\dot x(t)=f(x(t))+g(x(t))\,\mathbf{k}^{\star}_{\cA(t)}(t,x(t)),\quad x(t_0)=x_0.
\end{equation}
A solution of it on an interval ${I \subseteq [t_0,\infty)}$ containing $t_0$ is a continuous function ${x : I \to \R^n}$, $C^1$ on each $[t_j,t_{j+1})$, that satisfies it at every ${t \in I}$ with the right derivative at $t_j$.

\begin{corollary}\label{cor:append}
Consider the system \eqref{eq:system}. Suppose $S_0$ is a base set and ${t \mapsto \cA(t)}$ is an append path. Then, every solution of \eqref{eq:append:loop} with ${x(t_0) \in \OT(\mathcal R_{\cA_0})}$ satisfies ${x(t) \in \OT(\mathcal R_{\cA(t)}) \subseteq \Cs}$ on its domain, with ${m_{\cA(t)}(x(t)) > 0}$ if ${m_{\cA_0}(x(t_0)) > 0}$. Moreover, from every $x_0$ with ${m_{\cA_0}(x_0) > 0}$, \eqref{eq:append:loop} has a solution, and every maximal solution is defined on $[t_0,\infty)$.
\end{corollary}

\begin{proof}
See Appendix~\ref{app:append}.
\end{proof}

Corollary~\ref{cor:append} takes the append path as given, whereas online append decisions depend on the observed trajectory.

\subsection{Append Rules}

An \emph{append rule} decides at locally finite update times, from the trajectory up to that time and the family $\mathcal B$ in force, to retain $\mathcal B$ or to append finitely many links. It is \emph{certified} if every accepted append is an append extension of $\mathcal B$.

The rule fixes the family on each interval between update times from the trajectory up to its start; therefore, the induction in the proof of Corollary~\ref{cor:append} applies to the realized append path. Hence both claims of Corollary~\ref{cor:append} hold for a certified append rule started from a ball family $\cA_0$.

\emph{Demand-driven growth} is a certified rule whose search seeds are at the endpoints of rollouts taken along the nominal motion. At $t_j=jT$ it integrates the nominal closed loop $\dot{\hat x}=f(\hat x)+g(\hat x)\kp(t_j+s,\hat x)$ from $x(t_j)$ over $[0,T]$ by a fixed-step integrator and calls a rollout endpoint $q\defeq\phib(T,\hat x(s))$ a \emph{coverage deficit} of $\mathcal B$ when $H_{\mathcal B}(q)<0$. It visits the current deficits, earliest $s$ first, and then those carried from the previous update, and starts no visit after a computation time of $T$. Each visit works with the family $\mathcal B'$ formed by $\mathcal B$ and the links accepted so far in this update, and passes over a deficit in $\mathcal R_{\mathcal B'}$. Otherwise it follows the backup arc in blocks of duration at most $T$ measured from $q$, until the arc enters the interior of some $S_a$, $a\in\mathcal B'$, or a fixed budget of blocks is spent, cuts links at the block boundaries and docks the last link into $\mathcal B'$. The backup arc from $q$ and these blocks play the role of $c$ and of the segments of Theorem~\ref{thm:link}. The resulting links are accepted only if $\mathcal B'$ with them is an append extension of $\mathcal B'$. A deficit whose arc yields no links, or that is not visited, is carried to the next update.

\section{Planar Quadrotor Example}
\label{sec:sim}
\begin{figure}[tb]
\centering
\setlength{\abovecaptionskip}{1pt}
\includegraphics[width=\columnwidth]{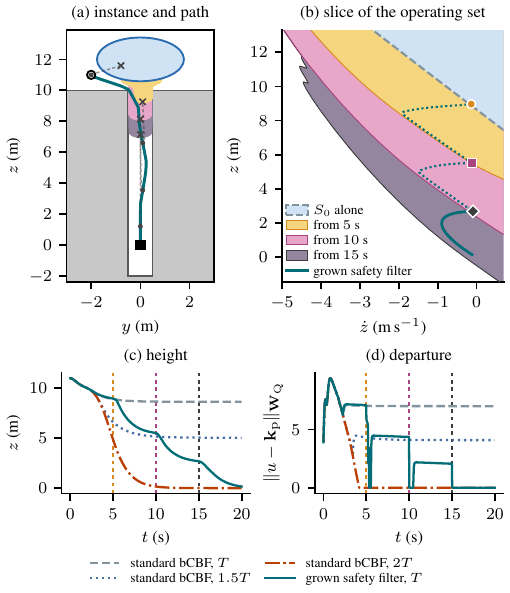}
\caption{Simulation results for a planar quadrotor descending through a narrow slot, comparing the grown safety filter (solid) with horizon $T$ and the standard bCBF with horizons $T$, $1.5T$, and $2T$.
(a) The trajectory under the grown safety filter, with the unsafe set in gray, the $y$--$z$ projection of the base set $S_0$ at the backup hover, and the stored balls colored as in (b). The initial position and goal are marked by $\circ$ and $\blacksquare$. Four backup predictions (dashed) start at dots along the trajectory and end at crosses.
(b) The operating set on the slice $y=\theta=\dot y=\dot\theta=0$ with $S_0$ alone and after each append. The trajectory is shown from the first append onward, with the portion before reaching the slice dotted and the append times marked. (c) The quadrotor height. (d) The control input deviation $\|u-\kp\|_{\WQ}$. Dashed vertical lines indicate the append times.}
\label{fig:sim}
\vspace{-4mm}
\end{figure}

We illustrate the grown safety filter on a planar quadrotor descending through a narrow slot.\renewcommand{\thefootnote}{\ddag}\footnote{The code and additional details of our simulations and results can be found at {\urlstyle{same}\url{https://github.com/myusufuzun/bcbf-memory}}}
The state and the input are $\xi\defeq [y~z~\theta~\dot y~\dot z~\dot\theta]^{\top}$, $u=[F~M]^{\top}$,
with the input set $\U=[0,12]\times[-12,12]$, and the dynamics
\begin{equation}
\label{eq:sim:model}
\!\!\ddot y \!={F\sin\theta}/{m_{\mathrm v}},\
\ddot z \!={F\cos\theta}/{m_{\mathrm v}}-9.81,\
\ddot \theta \!=-{M}/{J_{\mathrm v}},
\end{equation}
where $y$ and $z$ are the horizontal and the vertical position, $\theta$ the pitch, $F$ the thrust, $M$ the moment, $m_{\mathrm v}=1$ the mass and $J_{\mathrm v}=0.25$ the inertia. We use SI units.

The drone starts at rest at $(y,z)=(-2,11)$ and targets the origin, inside a slot of half-width $0.5$. The functions $h^{(1)}(\xi)=y+0.5$, $h^{(2)}(\xi)=0.5-y$, $h^{(3)}(\xi)=z+2$ and $h^{(4)}(\xi)=z-10$ describe the slot walls and floor. We take ${\nh = 1}$ and define $h_1$ as the smooth underapproximation \cite{molnar2025} of $\max\{\min\{h^{(1)},h^{(2)},h^{(3)}\},h^{(4)}\}$. As $h_1$ is strictly increasing in $z$, ${\operatorname{int}\Cs = \{\xi \in \R^6 : h_1(\xi) > 0\}}$.

The nominal controller $\kp$ and backup controller $\kb$ use the same cascaded PD structure with smooth saturation. From a target $(y_{\mathrm d},z_{\mathrm d})$ and gains $\gamma_{\mathrm p},\gamma_{\mathrm d}>0$ we form the commanded acceleration $a\defeq(a_y,a_z)$, with $a_y\defeq-\gamma_{\mathrm p}(y-y_{\mathrm d})-\gamma_{\mathrm d}\dot y$ and $a_z\defeq-\gamma_{\mathrm p}(z-z_{\mathrm d})-\gamma_{\mathrm d}\dot z+9.81$, and the law returns the thrust and the moment of \eqref{eq:sim:model} as
\begin{equation}\label{eq:sim:backup}
\begin{aligned}
    \theta_{\mathrm d}&=\operatorname{sat}_{[-\pi/4,\pi/4]}
    \bigl(\operatorname{atan2}(a_y,a_z)\bigr),\\
    F&=\operatorname{sat}_{[0,12]}\bigl(m_{\mathrm v}(a_y^2+a_z^2)^{1/2}\bigr),\\
    M&=\operatorname{sat}_{[-12,12]}
    \bigl(J_{\mathrm v}\,(64(\theta-\theta_{\mathrm d})+16\dot\theta)\bigr).
\end{aligned}
\end{equation}
The backup controller $\kb$ takes $(\gamma_{\mathrm p},\gamma_{\mathrm d})=(0.16,0.8)$ and the hover target $(0,12)$. The nominal controller $\kp$ takes $(0.81,1.8)$ and the goal $(0,0)$. Both laws \eqref{eq:sim:backup} are $C^\infty$ where $a_z>0$, so we take $\kb$ in \eqref{eq:setting:backup} and $\kp$ to be the laws with $a_z$ replaced by $\psi(a_z)$, where $\psi$ is $C^\infty$ and positive with ${\psi(\zeta) = \zeta}$ for ${\zeta \geq \underline{a}}$ and a fixed ${\underline{a} > 0}$.

The matrix $\mathbf{P}$ of the norm $\|\cdot\|_{\mathbf{P}}$ is fixed, as in \cite{gurriet2020}, by a Lyapunov ellipsoid of the backup closed loop linearized at the hover, with Jacobian $\mathbf{A}$. Its shape $\mathbf{W}$ has the largest volume with $\mathbf{A}\mathbf{W}+\mathbf{W}\mathbf{A}^{\top}+\nu \mathbf{W}\preceq0$ for some $\nu>0$, inside the safe set and the smooth actuator regime, halved in level so that the nonlinear backup field still contracts it, and $\mathbf{P}=2R_0^{2}\mathbf{W}^{-1}$ makes that ellipsoid the base set, the ball $S_0=B(\xi_\star,R_0)$ about the hover equilibrium $\xi_\star$ with $R_0=\sqrt{0.22}$. Its defining function $b_0$ is the ball function at $c=\xi_\star$ and $r=R_0$, and we take $\alpha_0(\zeta)=\alb(R_0\zeta)/R_0$. For $\kb$, this $S_0$ lies in ${\{a_z \geq \underline{a}\}}$, and we assume that this $b_0$ defines a base set. The prediction horizon is ${T = 5}$~s, and the remaining data are ${\alpha_1(\zeta) = 2\zeta + \zeta^3}$, ${\alb(\zeta) = 10\zeta + \zeta^3}$, ${\rho(\zeta) = 40\zeta}$ and ${\WQ = \mathbf{I}}$. Fig.~\ref{fig:sim}(a) shows the simulation setup and the trajectory under the grown safety filter.

We compare the grown safety filter with horizon $T$ against the \emph{standard bCBF} \cite{chen2021} with terminal set $S_0$ and horizons $T$, $1.5T$, and $2T$. Each run lasts $20$~s and uses RK4 with a $20$~ms step, with inputs held constant between steps. Both filters sample the rollout constraints every $20$~ms. Constraint tightening can certify safety between rollout samples \cite{vanwijk2024}. The grown family starts from $\cA_0=\{0\}$. Demand-driven growth searches start at $0$, $5$, and $10$~s, and accepted links are appended $T$ seconds later.

\looseness=-1
Fig.~\ref{fig:sim} shows that the quadrotor reaches the goal safely under the grown safety filter with $T=5$~s. Under the standard bCBF, it stops at $z = 8.64$ for horizon $T$, at $z = 5.04$ for $1.5T$, and reaches the goal with $2T$. On the slice in Fig.~\ref{fig:sim}(b), the three appends lower the operating set boundary at rest from $z = 8.64$ to $z = 5.33$, $2.36$, and $-0.49$. Each append permits further descent, and the control input equals $\kp$ from $15.0$~s onward. This terminal set growth increases the computation time per step, excluding the rollout, by at most $0.5$~ms, from $0.6$~ms with $S_0$ alone to $1.1$~ms with $6229$ stored balls. The median rollout computation time is $1.2$~ms for $T$ and $2.5$~ms for $2T$.

\section{Conclusion}
\label{sec:conclusion}
We presented a construction that grows the terminal set of a bCBF-based safety filter by appending certified links. We showed that the terminal set is forward invariant under the backup flow and that its states reach the base set within a bounded chain time that can exceed the horizon. The safety filter is feasible at every state and preserves safety under certified online appends. Future work will focus on bounding the family size and on designing searches guaranteed to find states from which links can be appended.

\appendix

\subsection{Proof of Lemma~\ref{lem:collar}}\label{app:collar}

\begin{proof}
By \eqref{eq:collar:eps}, ${0 < \epsA \leq \varepsilon_0}$. Fix $x$ with ${-\epsA < \HA(x) < \varepsilon_0}$. If ${0 \in \MA(x)}$, then ${|b_0(x)| = |\HA(x)| < \varepsilon_0}$, and ${\lamA \geq \lambda_0}$ with \eqref{eq:setting:basedrift} gives $\nabla b_0(x)\,\vb(x) \geq -\lambda_0\,|b_0(x)| \geq -\lamA\,|\HA(x)|$.

Suppose ${0 \notin \MA(x)}$. No base dock $a$ lies in $\MA(x)$, since ${\beta_a(x) = \HA(x) > -\epsA}$ implies ${\|x - c_a\| < R_a + \eta_a}$ and hence, by the docks condition (see Definition~\ref{def:semfam}), ${b_0(x) \geq \beta_a(x) = \HA(x)}$, that is, ${0 \in \MA(x)}$. A dock ${a \in \MA(x)}$ therefore has $p(a)$ in an older link, and ${\beta_{p(a)} \geq \beta_a}$ by the docks condition and \eqref{eq:env:nest}, so ${p(a) \in \MA(x)}$. Hence every ${(k,l) \in \MA(x)}$ with $k$ minimal is a flow step.

For such an $a$, let ${e \defeq x - c_a}$. Then ${\beta_a(x) = \HA(x)}$ gives
	\begin{equation*}
		\|e\|^2 = R_a^2 - R_a \HA(x) < R_a^2 + \epsA R_a \leq (R_a + \eta_a)^2 ,
	\end{equation*}
	so \eqref{eq:succ:osl} applies at $x$. Moreover, expanding ${\beta_{p(a)}(x) \leq \HA(x)}$ with ${x - c_{p(a)} = e - d_a}$ and the value of $\|e\|^2$ gives
	\begin{equation*}
		-2 e^\top \mathbf{P} d_a \geq R_{p(a)}^2 - R_a^2 - \|d_a\|^2 - \bigl(R_{p(a)} - R_a\bigr) \HA(x) ,
	\end{equation*}
	and ${\HA(x) \leq R_a}$, as ${\|e\|^2 \geq 0}$, gives ${\|e\| \leq R_a - \HA(x)/2}$. Hence, by \eqref{eq:succ:osl}, Cauchy--Schwarz, \eqref{eq:collar:const}, \eqref{eq:succ:exposure} and \eqref{eq:collar:lam},
	\begin{align*}
    	& \vartheta_a R_a \nabla \beta_a(x)\, \vb(x) \\
    	& = -2 e^\top \mathbf{P} \bigl[\vartheta_a \bigl(\vb(x) - \vb(c_a)\bigr) + \bigl(\vartheta_a \vb(c_a) - d_a\bigr) + d_a\bigr] \\
    	& \geq -2 \mu_a \vartheta_a \|e\|^2 - 2 \|e\| E_a - 2 e^\top \mathbf{P} d_a \\
    	& \geq \vartheta_a R_a \ell_a \HA(x) + R_{p(a)}^2 - (1 + 2 \mu_a \vartheta_a) R_a^2 - \|d_a\|^2 - 2 R_a E_a \\
    	& \geq \vartheta_a R_a \ell_a \HA(x) \geq -\vartheta_a R_a \lamA\, |\HA(x)| .\qedhere
    \end{align*}
\end{proof}

\subsection{Proof of Theorem~\ref{thm:growth}}\label{app:growth}

\begin{proof}
By Definitions~\ref{def:base} and~\ref{def:semfam}, ${\emptyset \neq S_0 \subseteq \RA \subseteq \Cs}$, and $\RA$ is compact as a finite union of compact sets.

For ${x \in \RA}$, the $C^1$ functions ${w_a(t) \defeq \beta_a(\phib(t,x))}$, ${a \in \cA}$, have the maximum ${w(t) = \HA(\phib(t,x))}$, with ${w(0) \geq 0}$ by \eqref{eq:env:level}. At every $t$ with ${w(t) \in (-\epsA, 0]}$, \eqref{eq:setting:complete} and Lemma~\ref{lem:collar} give ${\dot w_a(t) \geq -\lamA\,|w(t)| = \lamA\,w(t)}$ for some $a$ with ${w_a(t) = w(t)}$. Lemma~\ref{lem:nocross} with ${\varepsilon=\epsA}$ and ${\lambda=\lamA}$ on every $[0,t_1]$ then gives ${w(t) \geq 0}$, that is, ${\phib(t,x) \in \RA \subseteq \Cs}$, for all ${t \geq 0}$. Hence ${x \in \OT(\RA)}$. For the chain time claim, ${\phib(t, S_a) \subseteq S_0}$ for ${t \geq \Tch(a)}$ follows by induction along successor chains. For ${a = 0}$, this is the forward invariance of $S_0$. For ${a \neq 0}$, ${\phib(\vartheta_a,S_a)\subseteq S_{p(a)}}$ by the arrival condition (see Definition~\ref{def:semfam}) at a flow step and the docks condition at a dock, where ${\vartheta_a = 0}$, with ${b_0 \geq \beta_a \geq 0}$ on $S_a$ at a base dock. Then, for ${t \geq \Tch(a) = \vartheta_a+\Tch(p(a))}$, \eqref{eq:setting:semigroup} and the claim for $p(a)$ give ${\phib(t, S_a) \subseteq \phib(t - \vartheta_a, S_{p(a)}) \subseteq S_0}$. A successor chain crosses at most $\nlink$ links, each with ${L_k \leq T}$, so ${\Tch(a) \leq \nlink T}$.

For the append-monotonicity claim, ${\cA \subseteq \cA'}$ with the same balls on $\cA$ gives ${H_{\cA'} \geq \HA}$, hence ${\RA \subseteq \mathcal R_{\cA'}}$ by \eqref{eq:env:level} and ${\OT(\RA) \subseteq \OT(\mathcal R_{\cA'})}$ by \eqref{eq:setting:opset}.
\end{proof}

\subsection{Proof of Theorem~\ref{thm:link}}\label{app:link}

\begin{proof} The distances $\delta_L \defeq \dist(c(L), \R^n \setminus \operatorname{int} S_{a_L})$ and $\delta_{\mathrm s} \defeq \dist(c([0,L]), \R^n \setminus \operatorname{int} \Cs)$ are positive, with ${\delta_{\mathrm s} = \infty}$ if ${\Cs = \R^n}$. We take $r_{\max} \in (0, \min\{\delta_L, \delta_{\mathrm s}\})$ with ${b_0 \geq r_{\max}}$ on ${B(c(L), 2 r_{\max})}$ if ${a_L = 0}$, which is possible as $b_0$ is continuous and positive at ${c(L) \in \operatorname{int} S_0}$. Let $\mathcal N$ be the convex hull of the closed $2r_{\max}$-thickening of the arc and $\mu$ the maximum over the compact $\mathcal N$ of the logarithmic norm \cite{soderlind2006} of the Jacobian $\partial_x\vb$ in the norm $\|\cdot\|$, a one-sided Lipschitz constant of $\vb$ on $\mathcal N$ \cite{aminzare2014}. We then fix $\kappa\defeq\max\{\mu,0\}+1/L$, $\gamma_r\defeq\kappa-\mu>0$, $r(s)\defeq r_{\max} e^{\kappa(s-L)}$, $V_c\defeq\max_{[0,L]}\|\dot c\|$, $A_c\defeq\max_{[0,L]}\|\ddot c\|$ and
\begin{equation}\label{eq:link:count}
    N_0\defeq1+\biggl\lfloor \frac{L(V_c^2+r_{\max}A_c)}{2n_{\mathrm{s}}\gamma_r r(0)^2} \biggr\rfloor,
\end{equation}
where $\lfloor\cdot\rfloor$ denotes the floor function. Let ${N \geq N_0}$ and take a flow step $a$ of the new links, with ${\mu_a \defeq \mu}$ and ${\eta_a \defeq r_{\max}}$. Then ${\vartheta_a = L/(n_{\mathrm{s}}N)}$ and, for some ${s \in [0, L - \vartheta_a]}$, ${c_a = c(s)}$, ${R_a = r(s)}$, ${c_{p(a)} = c(s + \vartheta_a)}$ and ${R_{p(a)} = e^{\kappa\vartheta_a} R_a \leq r_{\max}}$.

For the safety condition (see Definition~\ref{def:semfam}), every new ball lies within ${r_{\max} < \delta_{\mathrm s}}$ of the arc, hence in $\Cs$. For the arrival condition (see Definition~\ref{def:semfam}), let ${x \in S_a}$. For ${t \in [0,\vartheta_a]}$ with ${\phib([0,t],x) \subseteq \mathcal N}$, the one-sided Lipschitz estimate for the solutions $\phib(\cdot,x)$ and $c(s+\cdot)$ of \eqref{eq:setting:backup} in the convex $\mathcal N$ gives ${\|\phib(t,x) - c(s + t)\| \leq e^{\mu t} R_a}$, and ${e^{\mu t} R_a \leq e^{\kappa t} R_a = r(s + t) \leq r_{\max}}$. This bound keeps $\phib(t,x)$ within $r_{\max}$ of the arc, hence in $\operatorname{int} \mathcal N$, on ${[0,\vartheta_a]}$, and at ${t = \vartheta_a}$ it gives ${\phib(\vartheta_a, S_a) \subseteq S_{p(a)}}$.

For the exposure condition (see Definition~\ref{def:semfam}), \eqref{eq:succ:osl} holds, as $\mathcal N$ contains ${B(c_a, R_a + \eta_a)}$. As $\vb$ is $C^1$, the arc $c$ is $C^2$, and Taylor's formula gives ${\|d_a\| \leq V_c \vartheta_a}$ and ${E_a \leq \tfrac12 A_c \vartheta_a^2}$. With ${e^{2\kappa\vartheta_a} \geq 1 + 2\kappa\vartheta_a}$, ${r(0) \leq R_a \leq r_{\max}}$ and \eqref{eq:link:count},
	\begin{equation*}
		\begin{aligned}
			& R_{p(a)}^2 - (1 + 2 \mu_a \vartheta_a) R_a^2 - \|d_a\|^2 - 2 R_a E_a \\
			& \quad \geq R_a^2 \bigl(e^{2 \kappa \vartheta_a} - 1 - 2 \mu \vartheta_a\bigr) - \bigl(V_c^2 + r_{\max} A_c\bigr) \vartheta_a^2 \\
			& \quad \geq \vartheta_a \bigl[2 \gamma_r r(0)^2 - \bigl(V_c^2 + r_{\max} A_c\bigr) \vartheta_a\bigr] > 0 ,
		\end{aligned}
	\end{equation*}
	which gives \eqref{eq:succ:exposure}.

For the docks condition (see Definition~\ref{def:semfam}), a dock $a$ where two segments meet has ${S_a = S_{p(a)}}$. As ${r_{\max} < \delta_L}$, the last dock $a$ has ${S_a = B(c(L),r_{\max}) \subseteq \operatorname{int}S_{a_L}}$, and, if ${a_L = 0}$, ${\eta_a \defeq r_{\max}}$ serves, as ${\beta_a \leq r_{\max} \leq b_0}$ on ${B(c(L), 2 r_{\max})}$.

For the horizon condition, each new link has duration ${L_k = L/n_{\mathrm{s}} \in (0,T]}$. The append order makes the successor of every dock older and leaves the old family unchanged, so the result is an append extension. The entrance at ${s = 0}$ contains $c(0)$, which gives strict enlargement if ${c(0) \notin \mathcal R_{\cA_{\mathrm{old}}}}$.
\end{proof}

\subsection{Proof of Theorem~\ref{thm:invariance}}\label{app:invariance}

\begin{proof}
By \eqref{eq:env:level},
\begin{equation}\label{eq:filter:opmargin}
	\OT(\RA)=\{x\in\R^n : m_{\cA}(x)\ge0\} .
\end{equation}
First, we prove the invariance and positive margin claims for a solution $x$ of \eqref{eq:inv:qploop} on ${I = [t_0,t_1]}$ or ${I = [t_0,t_1)}$ with ${x(t_0) \in \OT(\RA)}$. We fix ${\underline{m} \in [0, \varepsilon_0)}$ with ${\underline{m} \leq m_{\cA}(x(t_0))}$, taking ${\underline{m} = 0}$ for the invariance claim. As the gains are locally Lipschitz, there is ${\omega \geq \lamA}$ with ${\alpha_i(\zeta) \leq \omega\zeta}$ and ${\albhat_a(\zeta) \leq (\omega - \lamA)\zeta}$ for all $i$, ${a \in \cA}$ and ${\zeta \in [0, \varepsilon_0]}$. We begin with the terminal part. The function ${\gamma(t) \defeq \HA(x_T(t))}$, with ${x_T(t) = \phib(T, x(t))}$ along the solution, is the maximum of the $C^1$ functions ${t \mapsto \beta_a(x_T(t))}$, ${a \in \cA}$. At every $t$, $\dot x=\vb(x)+g(x)(u-\kb(x))$ with $u=\kQ(t,x)\in\KA(x)$, and \eqref{eq:setting:flowid} gives $\dot x_T=\vb(x_T)+\Phib(T,x)g(x)(u-\kb(x))$. Where ${\gamma \in (-\epsA, \underline{m}]}$, Lemma~\ref{lem:collar}, \eqref{eq:filter:term} with ${\rho(0) = 0}$ and the bound on $\albhat_a$ give, for some ${a \in \MA(x_T)}$,
\begin{equation*}
    \begin{aligned}
        \frac{d}{dt}\beta_a(x_T) &= \nabla\beta_a(x_T)\dot x_T \geq -\albhat_a(\gamma^{+}) +\nabla\beta_a(x_T)\vb(x_T) \\
        &\geq -(\omega - \lamA)\gamma^{+} - \lamA|\gamma| \geq -\omega|\gamma| .
    \end{aligned}
\end{equation*}
The functions ${e^{\omega (t - t_0)} \beta_a(x_T) - \underline{m}}$, ${a \in \cA}$, have the maximum ${e^{\omega (t - t_0)} \gamma - \underline{m}}$, which is nonnegative at $t_0$ and in ${(-\epsA, 0]}$ only where ${\gamma \in (-\epsA, \underline{m}]}$, so Lemma~\ref{lem:nocross} with $\varepsilon=\epsA$ and ${\lambda = 2\omega}$ on every ${[t_0,t] \subseteq I}$ gives ${\gamma(t) \geq \underline{m}\, e^{-\omega (t - t_0)}}$.

For the rollout part, we fix $i$ and ${(t^\ast, \tau^\ast) \in I \times [0, T]}$ and set ${t_\natural \defeq \max\{t_0, t^\ast - (T - \tau^\ast)\}}$. For ${t \in [t_\natural, t^\ast]}$, we let ${\tau \defeq \tau^\ast + t^\ast - t \in [\tau^\ast, T]}$, ${x_\tau \defeq \phib(\tau, x(t))}$ and ${w_i(t) \defeq h_i(x_\tau)}$, a $C^1$ function. Then \eqref{eq:setting:complete}, \eqref{eq:setting:flowid} and \eqref{eq:filter:roll} at $\tau$ give
\begin{equation*}
    \begin{aligned}
        \dot w_i&=\nabla h_i(x_\tau)\bigl[\Phib(\tau,x)\dot x-\vb(x_\tau)\bigr]\\
        &=\nabla h_i(x_\tau)\Phib(\tau,x)g(x)\bigl(u-\kb(x)\bigr)
        \ge-\alpha_i\bigl(w_i^{+}\bigr).
    \end{aligned}
\end{equation*}
The set ${D \defeq \{\HA \geq \underline{m}\, e^{-\omega (t_1 - t_0)}\}}$ contains $x_T(t_0)$ and is a compact subset of ${\RA \subseteq \Cs}$, by Theorem~\ref{thm:growth}, and, if ${\underline{m} > 0}$, of the open set ${\{\HA > 0\} \subseteq \operatorname{int} \Cs}$, on which every $h_j$ is positive. Hence ${\varpi \defeq \min\{\underline{m}, \min_{j} \min_{D} h_j\}}$ satisfies ${0 \leq \varpi \leq \underline{m}}$, with ${\varpi > 0}$ if ${\underline{m} > 0}$. We have ${w_i(t_\natural) \geq \varpi}$, by ${m_{\cA}(x(t_0)) \geq \underline{m}}$ if ${t_\natural = t_0}$ and otherwise because ${\tau = T}$ there and ${x_T(t_\natural) \in D}$ by the terminal part. Where ${w_i \leq \varpi}$, ${\dot w_i \geq -\omega w_i^{+}}$ by the bound on $\alpha_i$, so Lemma~\ref{lem:nocross} with the single function ${e^{\omega (t - t_\natural)} w_i - \varpi}$, any $\varepsilon>0$ and ${\lambda = \omega}$ gives ${h_i(\phib(\tau^\ast,x(t^\ast))) = w_i(t^\ast) \geq \varpi e^{-\omega (t_1 - t_0)}}$. As $i$, $t^\ast$ and $\tau^\ast$ are arbitrary and ${\varpi \leq \underline{m}}$, the terminal part gives ${m_{\cA}(x(t)) \geq \varpi e^{-\omega (t_1 - t_0)}}$ for all ${t \in I}$, which is the positive margin claim if ${\underline{m} > 0}$ and, by \eqref{eq:filter:opmargin}, the invariance claim if ${\underline{m} = 0}$. The inclusion ${\OT(\RA) \subseteq \Cs}$ is from \eqref{eq:setting:opset} with ${\tau = 0}$.

For the continuity claim, let $\chi(\iota,x,u)$ be the left minus the right side of constraint $\iota$, the constraints being indexed by the compact set ${J\defeq(\{1,\dots,\nh\}\times[0,T])\sqcup\cA}$. Then $\chi$ is affine in $u$ and continuous on ${J\times\R^n\times\U}$, as $\HA$ is continuous, so $\KA$, with ${\KA(x) = \{u \in \U : \min_{\iota \in J} \chi(\iota,x,u) \geq 0\}}$, has a closed graph and, $\U$ being compact, is upper semicontinuous. Let ${m_{\cA}(x) > 0}$. Then ${\underline{\chi} \defeq \min_{\iota \in J} \chi(\iota,x,\kb(x)) > 0}$, as the gains read arguments of at least $m_{\cA}(x)$ and ${\rho \geq 0}$. For ${u_0 \in \KA(x)}$ the points ${u_\varsigma \defeq (1 - \varsigma) u_0 + \varsigma \kb(x) \in \U}$, ${\varsigma \in (0,1]}$, have ${\chi(\iota,x,u_\varsigma) \geq \varsigma \underline{\chi}}$ for all $\iota$, hence stay feasible near $x$, and ${u_\varsigma \to u_0}$ as ${\varsigma \to 0}$. So $\KA$ is lower semicontinuous at $x$, and Berge's theorem and strict convexity make $\kQ$ continuous at $(t,x)$. For the closed loop claim, ${\{m_{\cA} > 0\}}$ is open, as $m_{\cA}$ is continuous, so the continuity claim and Peano's theorem give a local solution of \eqref{eq:inv:qploop} from every $x_0$ in it. By the positive margin claim and \eqref{eq:filter:opmargin}, a solution on a bounded interval $[t_0,t_1)$ from such an $x_0$ stays in a compact subset of ${\{m_{\cA} > 0\}}$, as $\OT(\RA)$ is compact, being closed and contained in $\phib(-T,\RA)$. With $f$ and $g$ bounded there and $\U$ compact, every maximal solution is global.
\end{proof}

\subsection{Proof of Corollary~\ref{cor:append}}\label{app:append}

\begin{proof}
	Every $\cA_j$ is a ball family, so Theorems~\ref{thm:growth} and~\ref{thm:invariance} apply to it. We first prove the persistence claim by induction over the update times, finitely many up to each $t$ in the domain. Suppose ${x(t_j) \in \OT(\mathcal R_{\cA_j})}$, which holds for ${j = 0}$. For every ${t \in (t_j,t_{j+1})}$ in the domain, $x$ solves \eqref{eq:inv:qploop} for $\cA_j$ from $x(t_j)$ on $[t_j,t]$, so the invariance claim of Theorem~\ref{thm:invariance} gives ${x(t) \in \OT(\mathcal R_{\cA_j})}$. If $t_{j+1}$ lies in the domain, continuity of $x$, closedness of $\OT(\mathcal R_{\cA_j})$ and the append-monotonicity claim of Theorem~\ref{thm:growth} give ${x(t_{j+1}) \in \OT(\mathcal R_{\cA_j}) \subseteq \OT(\mathcal R_{\cA_{j+1}})}$. If also ${m_{\cA_j}(x(t_j)) > 0}$, the positive margin claim of Theorem~\ref{thm:invariance} on $[t_j,t)$ and continuity give ${m_{\cA_j}(x(t)) > 0}$ for every ${t \in (t_j,t_{j+1}]}$ in the domain, and ${m_{\cA_{j+1}} \geq m_{\cA_j}}$ by \eqref{eq:filter:margin}, as ${H_{\cA_{j+1}} \geq H_{\cA_j}}$. For the closed loop claim, let ${m_{\cA_0}(x_0) > 0}$. The closed loop claim of Theorem~\ref{thm:invariance} gives a solution $x$ of \eqref{eq:inv:qploop} for $\cA_0$ from $x_0$ on $[t_0,\infty)$. On $[t_0,t_1)$ it solves \eqref{eq:append:loop}, so, as in the induction, ${m_{\cA_1}(x(t_1)) > 0}$, the next solution starts from $x(t_1)$, and so on. The concatenation solves \eqref{eq:append:loop}, and local finiteness makes it defined on $[t_0,\infty)$. Every maximal solution is global, since on $[t_j,t_{j+1})$ it solves \eqref{eq:inv:qploop} for $\cA_j$ from a state with positive $m_{\cA_j}$, whose maximal solutions the closed loop claim of Theorem~\ref{thm:invariance} carries to $t_{j+1}$.
\end{proof}

\begin{spacing}{0.90}
\bibliographystyle{IEEEtran}
\bibliography{refs}
\end{spacing}

\end{document}